\documentclass[english]{article}
\usepackage[T1]{fontenc}
\usepackage[latin9]{inputenc}
\usepackage{geometry}
\usepackage{url}
\usepackage{amsmath}
\usepackage{amsthm}
\usepackage{amssymb}
\usepackage{graphicx}
\usepackage{setspace}
\makeatletter

\providecommand{\tabularnewline}{\\}

\theoremstyle{plain}
\newtheorem{thm}{\protect\theoremname}
\theoremstyle{definition}
\newtheorem{defn}[thm]{\protect\definitionname}
\theoremstyle{definition}
\newtheorem{example}[thm]{\protect\examplename}
\theoremstyle{plain}
\newtheorem{lem}[thm]{\protect\lemmaname}
\theoremstyle{plain}
\newtheorem{cor}[thm]{\protect\corollaryname}

\usepackage{indentfirst}

\usepackage{babel}
\providecommand{\corollaryname}{Corollary}
\providecommand{\definitionname}{Definition}
\providecommand{\examplename}{Example}
\providecommand{\lemmaname}{Lemma}
\providecommand{\theoremname}{Theorem}

\makeatother

\usepackage{babel}
\providecommand{\corollaryname}{Corollary}
\providecommand{\definitionname}{Definition}
\providecommand{\examplename}{Example}
\providecommand{\lemmaname}{Lemma}
\providecommand{\theoremname}{Theorem}

\begin{document}
\title{{\Large{}A General Framework for Impermanent Loss in Automated Market
Makers}}
\author{Neelesh Tiruviluamala\,\,\,\,\,Alexander Port\,\,\,\,\,Erik
Lewis\thanks{Emails: neel@thrackle.io, alex@thrackle.io, erik@thrackle.io}}
\maketitle
\begin{abstract}
We provide a framework for analyzing impermanent loss for general
Automated Market Makers (AMMs) and show that Geometric Mean Market
Makers (G3Ms) are in a rigorous sense the simplest class of AMMs from
an impermanent loss viewpoint. In this context, it becomes clear why
automated market makers like Curve (\cite{Ego19}) require more parameters
in order to specify impermanent loss. We suggest the proper parameter
space on which impermanent loss should be considered and prove results
that help in understanding the impermanent loss characteristics of
different AMMs. 
\end{abstract}

\section{Introduction}

Impermanent loss for protocols like Uniswap (\cite{Ada18}) and Balancer
(\cite{MM19}) have been well studied (\cite{Eva20}, \cite{AD21},
\cite{AEC20}, \cite{Aoy20}, \cite{Bou21}). These geometric mean
market maker protocols (G3Ms) can give the false impression that certain
non-trivial properties regarding impermanent loss should hold for
all Automated Market Makers (AMMs). For instance, Figure \ref{Fig 1}
implicitly leverages the idea that impermanent loss is a one free
parameter function in the case of two dimensional G3Ms. Not by design
by true nevertheless, Curve's StableSwap AMM protocol (\cite{Ego19})
does not allow impermanent loss to be analyzed with one free parameter
(so Curve could not be compared in Figure \ref{Fig 1}). At a high
level, the goal of this paper is to provide a useful way to analyze
impermanent loss as protocols continue to increase in complexity.

\begin{figure}[t]
\noindent \begin{centering}
\includegraphics[scale=0.65]{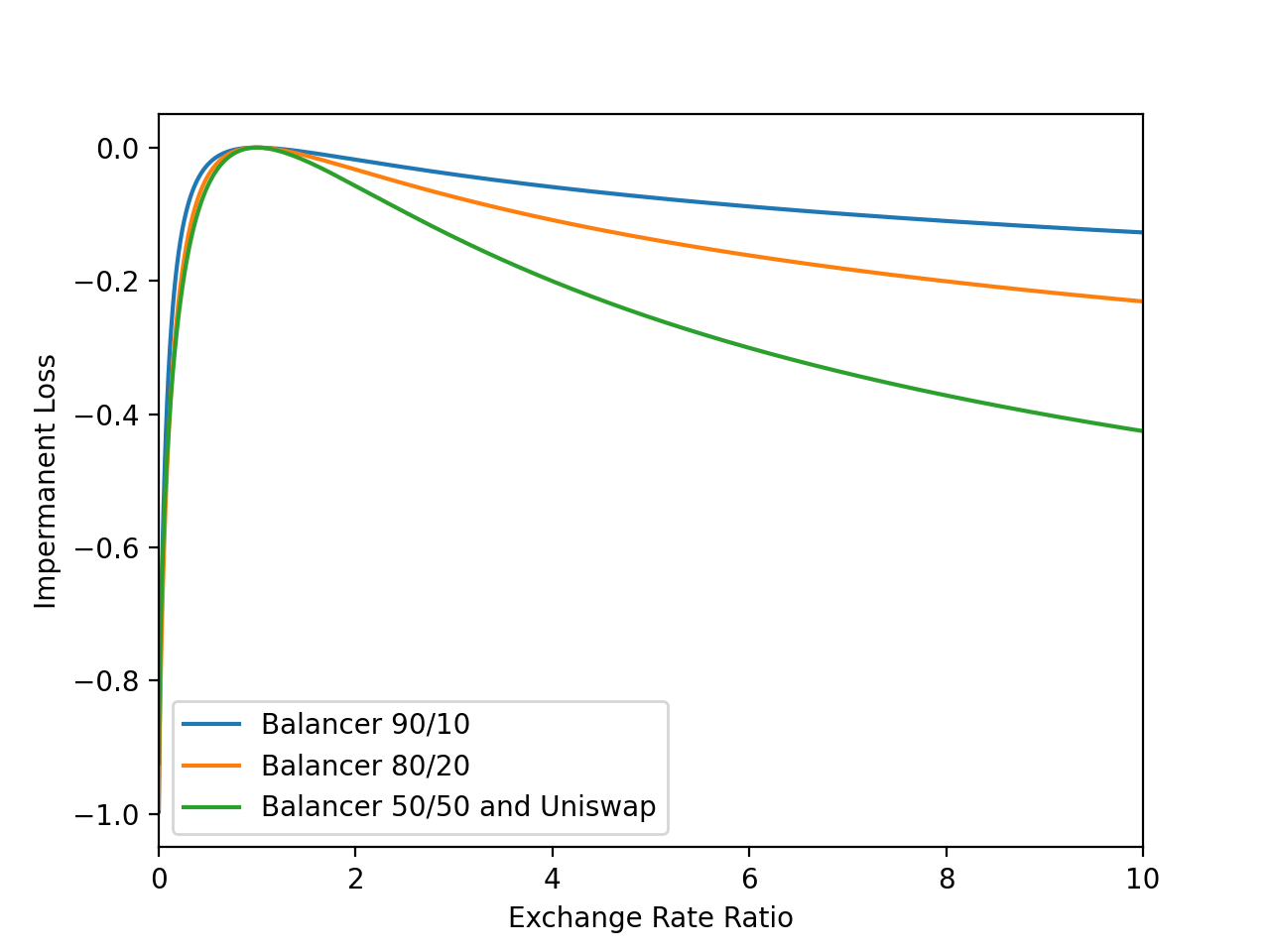} 
\par\end{centering}
\caption{\label{Fig 1}Impermanent Loss for Uniswap and Balancer Pools Compared}
\end{figure}

Along these lines, we start by observing and codifying some of the
properties of G3Ms. The major property that we focus on in this paper
is what we call \emph{Exchange Rate Level Independence} (ERLI). This
property states that impermanent loss can be understood entirely in
terms of ratios of initial and final exchange rates, as opposed to
the actual rates themselves. In other words, moving from an Eth $:$
BTC exchange rate of $15:1$ to an exchange rate of $30:1$ has the
same effect on impermanent loss as moving from an exchange rate of
$10:1$ to an exchange rate of $20:1$, since in both scenarios the
exchange rate has doubled. G3Ms exhibit this property, though to our
knowledge, it has not been clearly formalized for higher dimensional
AMMs until now. Among other conceptual benefits, exchange rate level
independence allows us to analyze impermanent loss for AMMs with fewer
parameters.

Showing that G3Ms (up to a simple transformation) have the ERLI property
and are the only AMMs to have this property is the main result of
this paper. The connection between portfolio value functions and trading
functions pointed out in \cite{AEC21} is a crucial first step towards
establishing this result. However, a slightly different viewpoint
is used in this paper. We focus on individual \emph{level surfaces}
of an AMM as opposed to the function $A$ that generates them. The
following example might motivate this shift in focus. Consider 
\begin{align*}
A(x,y) & =xy\\
B(x,y) & =e^{xy}
\end{align*}

\noindent Both of these functions have the same level curves, namely
the liquidity curves of Uniswap, but only the first function is a
geometric mean market maker. As such, trying to infer properties of
$A$ from its impermanent loss characteristics is misguided since
both the above market makers would lead to the same impermanent loss.
On the other hand, restricting a priori to the class of geometric
mean market makers precludes some interesting AMMs like Curve's StableSwap.
Focusing on the level surfaces of $A$ has conceptual benefits as
well, some of which are outlined in Section \ref{TheLegend}.

Readers familiar with the space can skip to Section $3$, in which
we derive an impermanent loss formula for higher dimensional constant
product market makers. In the process of doing so, we note that these
constant product market makers have some helpful properties that we
work to concretize in the later sections. We then show that the smallest
class of market makers with these properties is the class of market
makers whose level surfaces match those of a geometric mean market
maker.

\section{Background and Notation}

In what follows, we will use a definition for an automated market
maker (AMM) motivated by \cite{EH21}. 
\begin{defn}
\emph{An automated market maker is a map $A=A(x_{1},...,x_{n}):\mathbb{R}_{>0}^{n}\rightarrow\mathbb{R}$
where} 
\end{defn}

\begin{enumerate}
\item \emph{$A\in C^{2}(\mathbb{R}_{>0}^{n},\mathbb{R})$} 
\item \emph{$Im(A_{x_{i}})\subset\mathbb{R}_{>0}$ for all $i\in\{1,...,n\}$} 
\item \emph{$upper(A,k)=\{(x_{1},...,x_{n})\in\mathbb{R}_{>0}^{n}:A(x_{1},...,x_{n})\geq k\}$
is strictly convex for all $k\geq0$.} 
\end{enumerate}
\noindent Note that some authors require that the AMM is homogeneous
(i.e. $A(cx_{1},...,cx_{n})=c^{d}A(x_{1},...,x_{n})$ for some $d>0$)
while other authors don't; see \cite{CJ21} and \cite{EH21} respectively.
In practical terms, an AMM provides a programmatic way for traders
to swap tokens for one another. In the above definition, $n$ represents
the number of tokens, and $x_{i}$ is the quantity of token $i$.
For a given state of quantities $x=(x_{1},...,x_{n})\in\mathbb{R}_{>0}^{n}$,
there is a level surface (or curve if $n=2$) $S$ of $A$ that passes
through this state. Traders are allowed to swap tokens in and out
of the AMM in any manner that leaves the resulting state on the surface
$S$. In contrast, liquidity providers through their actions can move
the state from one level surface to another. The first condition ensures
smoothness of these surfaces, the second condition ensures that the
level surfaces are sensibly indexed, and the third condition ensures
that the level surfaces are convex. It is worth mentioning that the
commonly presented two-dimensional hyperbolic constant product market
makers correspond to the function $A(x_{1},x_{2})=x_{1}x_{2}$, some
of whose level curves are pictured below. For a more thorough introduction
to these ideas, consult \cite{AC20}.

\begin{figure}[t]
\noindent \begin{centering}
\includegraphics[scale=0.75]{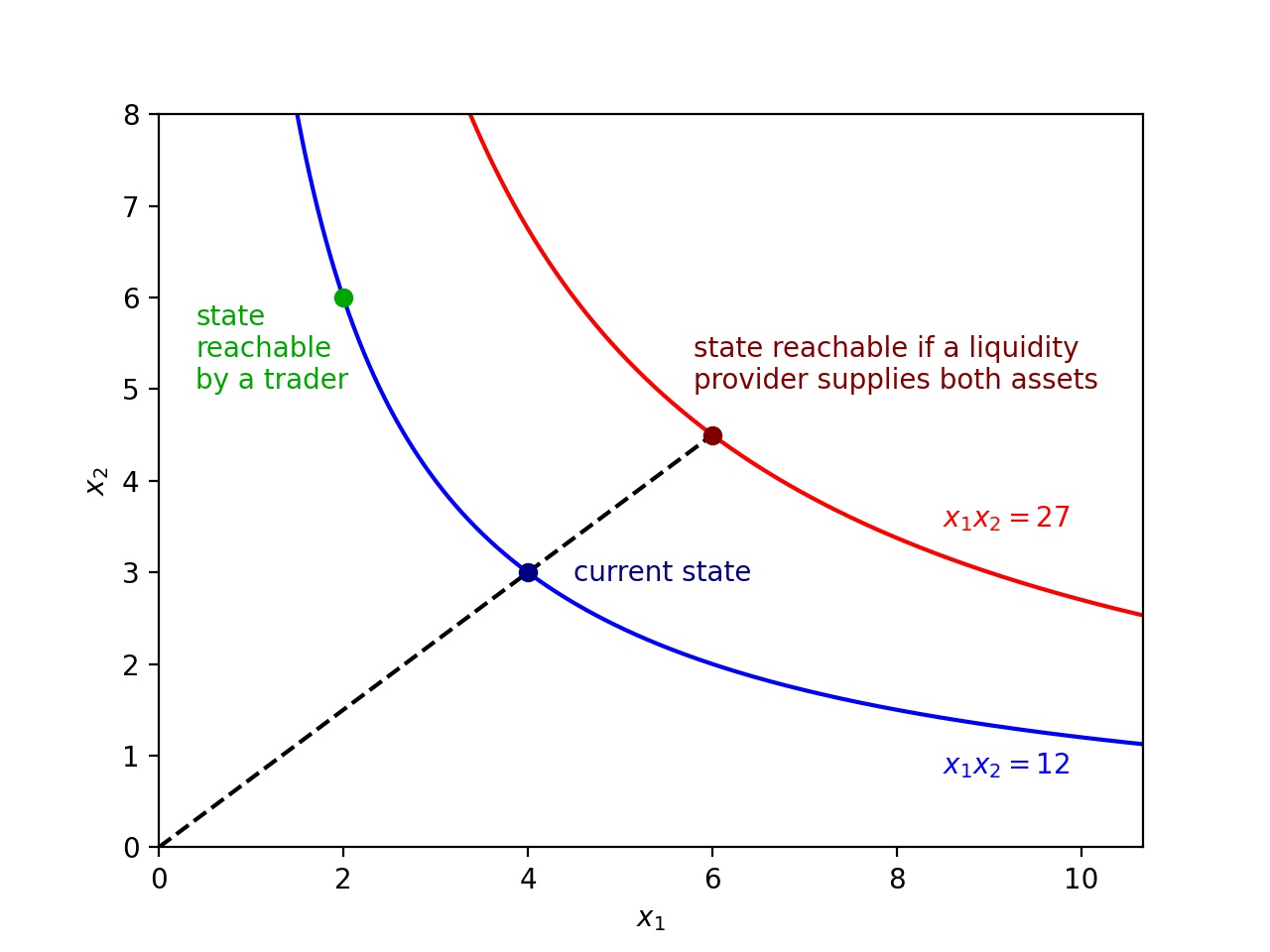} 
\par\end{centering}
\caption{\label{Fig 2}Two level curves of the AMM $A(x_{1},x_{2})=x_{1}x_{2}$
are pictured in state space. The current state $(4,3)$ consists of
4 of the first token type and $4$ of the second token type. A trader
could, for example, extract $2$ tokens of the first type in exchange
for supplying $3$ tokens of the second type. This would push the
state to $(4,3)+(-2,3)=(2,6)$, which satisfies the requirement that
it falls on the level curve $A=12$. Alternatively, a liquidity provider
could choose to add $2$ tokens of the first type and $1.5$ tokens
of the second type, expanding the pool's liquidity, but leaving its
reserves in the same ratio.}
\end{figure}

A price vector $p=(p_{1},...,p_{n})$ assigns a price to each token
so that $p_{i}$ is the price of token $i$ in terms of some measuring
currency. A price vector can be thought of in state space as the gradient
vector (it points in the direction of steepest increase) of the value
function 
\[
V(x)=p\cdot x=\sum_{i=1}^{n}p_{i}x_{i}
\]
This value function is naturally defined in that it is the sum of
the product of each token quantity by the corresponding token price.
If we think of $V$ as a function on the state space (given a fixed
price vector $p$), we see that $V$ has a constant gradient $p$.
We see in Figure \ref{Fig 3} the relationship between the price vector
$p$ and the level curves of the value function $V$ that it induces.

\begin{figure}[t]
\noindent \begin{centering}
\includegraphics[scale=0.75]{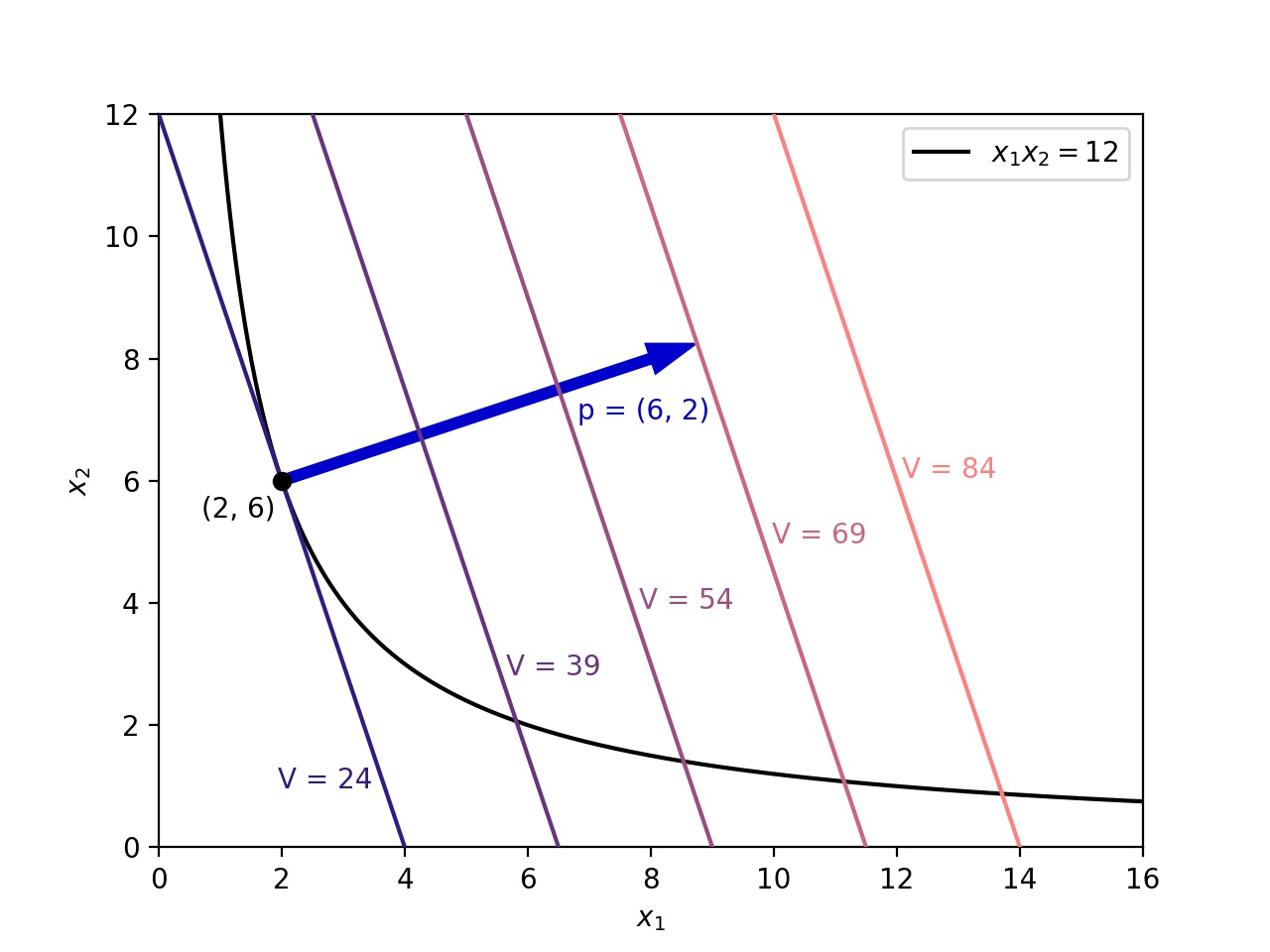} 
\par\end{centering}
\caption{\label{Fig 3}Several level curves of the value function $V$ corresponding
to the price vector $p=(6,2)$ are pictured in state space. Note that
the $V=24$ level curve is tangent to the AMM level curve $A=x_{1}x_{2}=12$
at the state $(2,6)$. It is geometrically clear, then, that the AMM
(for the $A=12$ level curve) obtains its lowest value at the state
$(2,6)$. As such, if the external market prices for the tokens were
given by $p$, arbitrageurs would drive the state to $(2,6)$.}
\end{figure}

In Figure \ref{Fig 3}, the state $(2,6)$ is referred to as a \emph{stable
point}, in that it is the state with the lowest value $V$ on the
current trading curve $A=12$. We will define stable points more formally
below. Note that the measuring currency is not important in the context
of defining stable points. Changing the measuring currency through
which $p$ is defined would stretch or contract $p$, but it would
not change its direction. Also, changing the measuring currency would
leave the family of level curves (surfaces) intact, though it would
alter the $V$ values corresponding to these level curves (surfaces).
If, for instance, the measuring currency were taken to be the first
token, $p_{1}$ would be $1$ and $p_{2}$ would be $\frac{2}{6}$,
which would result in the same level curves for the value function
as shown in Figure \ref{Fig 3}, but with the corresponding values
divided by $6$.

Let's make a fixed choice of price $p\in\mathbb{R}_{>0}^{n}$, AMM
$A:\mathbb{R}_{>0}^{n}\rightarrow\mathbb{R}$, and number $k\in\mathbb{R}$.
We say that a point $(x_{1}^{o},...,x_{n}^{o})\in\{(x_{1},...,x_{n})\in\mathbb{R}_{>0}^{n}:A(x_{1},...,x_{n})=k\}$
is \emph{stable} if

\[
\sum_{i=1}^{n}p_{i}x_{i}^{o}=inf\{\sum_{i=1}^{n}p_{i}x_{i}:(x_{1},...,x_{n})\in\mathbb{R}_{>0}^{n},A(x_{1},...,x_{n})=k\}
\]

\noindent In other words, a state $x^{o}$ is stable if it has the
lowest value $V(x^{o})$ among all states on the level surface that
it defines. This corresponds to the idea that arbitrageurs will drive
the AMM state to this stable state based on the external market valuation
of the assets. Finding such a stable point is a Lagrange multipliers
problem; we are trying to minimize the value of $V(x)$ subject to
the condition that $A(x)=k$. We know that $\nabla V=p$ and therefore
$p$ must be parallel to $\nabla A=(A_{x_{1}},...,A_{x_{n}})$ evaluated
at the given state. For instance, in Figure \ref{Fig 3} one can observe
that the gradient of $A$ at the stable point $(2,6)$ is parallel
to $p$.

In what follows, we assume that if the price vector changes (one can
think of it as fluctuating exogenously), the state vector in the AMM
changes to the corresponding stable point (thanks to arbitrageurs).
In \cite{EH21}, it is shown that for a given price vector $p$, there
is a unique stable point $x$ on the $A=k$ surface. In other words,
the stable state is a function of price

\noindent 
\[
x=x(p)
\]
but we will avoid using this notation to keep the exposition clean.
Suppose that the price vector changes from $p^{i}$ to $p^{f}$. The
following definition helps to quantify the loss liquidity providers
face when compared to those who simply hold their assets. It is commonly
referred to in the literature as \emph{impermanent loss} or \emph{divergence
loss}. 
\begin{defn}
\emph{Fix an AMM $A:\mathbb{R}_{>0}^{n}\rightarrow\mathbb{R}$ and
a level surface $A=k$. Suppose the initial and final prices are given
by $p^{i}=(p_{1}^{i},...,p_{n}^{i})$ and $p^{f}=(p_{1}^{f},...,p_{n}^{f})$,
and that the corresponding stable states are $x^{i}=(x_{1}^{i},...,x_{n}^{i})$
and $x^{f}=(x_{1}^{f},...,x_{n}^{f})$. Then the impermanent loss
is defined to be}

\noindent \emph{ 
\begin{eqnarray*}
IL & = & \frac{\text{Final value of pool assets}-\text{Value if assets were held}}{\text{Value if assets were held}}\\
 & = & \frac{V(x^{f})-V(x^{i})}{V(x^{i})}=\frac{p^{f}\cdot x^{f}-p^{f}\cdot x^{i}}{p^{f}\cdot x^{i}}=\frac{p^{f}\cdot x^{f}}{p^{f}\cdot x^{i}}-1
\end{eqnarray*}
}

Note that the above formula for impermanent loss is dependent on $p^{i}$
since $x^{i}$ is a function of $p^{i}$. Thus, impermanent loss at
first glance seems to be a function of $p^{i}$ and $p^{f}$, which
involves $2n$ parameters. In the next section, we will investigate
whether we can in general reduce the number of parameters needed to
compute impermanent loss. To provide context for this investigation,
we will conclude this section with a simple example. We will derive
the well known formula for the impermanent loss of a two dimensional
constant product market maker. 
\end{defn}

\begin{example}
\label{Ex 3}\emph{For a constant product market maker with two assets
$x_{1}$ and $x_{2}$ where $A(x_{1},x_{2})=x_{1}x_{2}=k$, consider
two assets with initial prices $p_{1}^{i}$ and $p_{2}^{i}$, initial
quantities $x_{1}^{i}$ and $x_{2}^{i}$, final prices $p_{1}^{f}$
and $p_{2}^{f}$, and final quantities $x_{1}^{f}$ and $x_{2}^{f}$.
The impermanent loss formula given by}

\noindent \emph{ 
\[
IL=\frac{p_{1}^{f}x_{1}^{f}+p_{2}^{f}x_{2}^{f}}{p_{1}^{f}x_{1}^{i}+p_{2}^{f}x_{2}^{i}}-1\text{\thinspace\thinspace can be reduced to\thinspace\thinspace}\frac{2\sqrt{t}}{t+1}-1\text{\thinspace\thinspace where\thinspace\thinspace}t=\frac{\frac{p_{2}^{f}}{p_{1}^{f}}}{\frac{p_{2}^{i}}{p_{1}^{i}}}
\]
} 
\end{example}

\begin{proof}
\emph{Before we prove the statement, we will provide intuition for
what $t$ signifies. We first define $m$ to be the exchange rate
between tokens $1$ and $2$, so that $m$ is the price of token $2$
in terms of token $1$. Concretely, $m_{f}=\frac{p_{2}^{f}}{p_{1}^{f}}$
and $m_{i}=\frac{p_{2}^{i}}{p_{1}^{i}}$. The parameter $t$ is the
quotient of the final and initial exchange rates. The further $t$
is away from $1$, the more drastically the exchange rate between
the tokens has changed.}

\emph{To prove the claim, we can start by expressing the quantities
in terms of prices. As noted above, with the assumption that arbitrageurs
will drive the value of the pool to a minimum, we can use Lagrange
multipliers to find the stable state $x$ as a function of $p$. We
see that $\nabla A=(x_{2},x_{1})$, so solving}

\[
\begin{array}{ccccc}
x_{2}=\lambda p_{1} &  & x_{1}=\lambda p_{2} &  & x_{1}x_{2}=k\end{array}
\]

\noindent \emph{yields $x_{1}=\sqrt{k}\sqrt{\frac{p_{2}}{p_{1}}}$
and $x_{2}=\sqrt{k}\sqrt{\frac{p_{1}}{p_{2}}}$. Substituting this
back into the original equation, we get}

\[
IL=\frac{p_{1}^{f}\sqrt{k}\sqrt{\frac{p_{2}^{f}}{p_{1}^{f}}}+p_{2}^{f}\sqrt{k}\sqrt{\frac{p_{1}^{f}}{p_{2}^{f}}}}{p_{1}^{f}\sqrt{k}\sqrt{\frac{p_{2}^{i}}{p_{1}^{i}}}+p_{2}^{f}\sqrt{k}\sqrt{\frac{p_{1}^{i}}{p_{2}^{i}}}}-1=\frac{2\sqrt{p_{1}^{f}p_{2}^{f}}}{p_{1}^{f}\sqrt{\frac{p_{2}^{i}}{p_{1}^{i}}}+p_{2}^{f}\sqrt{\frac{p_{1}^{i}}{p_{2}^{i}}}}-1
\]

\noindent \emph{Dividing the numerator and denominator by $\sqrt{p_{1}^{f}p_{2}^{f}}$,
we get}

\emph{ 
\[
IL=\frac{2}{\sqrt{\frac{p_{1}^{f}}{p_{2}^{f}}\sqrt{\frac{p_{2}^{i}}{p_{1}^{i}}}+\sqrt{\frac{p_{2}^{f}}{p_{1}^{f}}}\sqrt{\frac{p_{1}^{i}}{p_{2}^{i}}}}}-1
\]
}

\noindent \emph{Using the definition that $m=\frac{p_{2}}{p_{1}}$,
we obtain}

\emph{ 
\[
IL=\frac{2}{\frac{\sqrt{m_{i}}}{\sqrt{m_{f}}}+\frac{\sqrt{m_{f}}}{\sqrt{m_{i}}}}-1
\]
}

\noindent \emph{Finally, letting $t=\frac{m_{f}}{m_{i}}$, we arrive
at the desired result.}
\end{proof}
As is well known, the above formula demonstrates that impermanent
loss is least destructive when $t=1$, that is when the exchange rate
between tokens does not change (here $IL=0$). The key takeaway that
we will try to generalize in what follows is that impermanent loss
is best understood in terms of this one parameter $t$, as opposed
to the four parameters $p_{1}^{i}$, $p_{2}^{i}$, $p_{1}^{f}$ and
$p_{2}^{f}$.

\section{Higher Dimensional Constant Product Market Makers}

To gain an intuition for the parameters that should be used to conceptualize
impermanent loss in higher dimensions, we start by proving a theorem
involving an $n$ dimensional constant product market maker. 
\begin{thm}
Consider a constant product market maker with assets $x_{1},...,x_{n}$
where

\[
A(x_{1},...,x_{n})=x_{1}...x_{n}=k
\]

\noindent Let the initial prices of the assets be given by $p_{1}^{i},...,p_{n}^{i}$,
initial quantities $x_{1}^{i},...,x_{n}^{i}$, final prices $p_{1}^{f},...,p_{n}^{f}$,
and final quantities $x_{1}^{f},...,x_{n}^{f}$. The impermanent loss
formula given by

\[
IL=\frac{\sum_{j=1}^{n}p_{j}^{f}x_{j}^{f}}{\sum_{j=1}^{n}p_{j}^{f}x_{j}^{i}}-1\text{\thinspace\thinspace reduces to\thinspace\thinspace}\frac{n\prod_{j=2}^{n}t_{j}^{\frac{1}{n}}}{1+\sum_{j=2}^{n}t_{j}}-1\text{\thinspace\thinspace where\thinspace\thinspace}t_{j}=\frac{\frac{p_{j}^{f}}{p_{1}^{f}}}{\frac{p_{j}^{i}}{p_{1}^{i}}}
\]
\end{thm}

\begin{proof}
\emph{We again start by expressing the quantities in terms of prices.
To do so, similarly as in Example \ref{Ex 3}, we must solve the Lagrange
multipliers equations:}

\noindent \emph{
\begin{eqnarray*}
\prod_{j\neq i}x_{j} & = & \lambda p_{i}\text{\thinspace\thinspace for each\thinspace\thinspace\ \ensuremath{i\in\{1,...,n\}}}\\
\prod_{j=1}^{n}x_{j} & = & k
\end{eqnarray*}
}

\noindent \emph{Multiplying the first $n$ equations leads to the
equation}

\emph{ 
\begin{eqnarray*}
\prod_{j=1}^{n}x_{j}^{n-1} & = & \lambda^{n}\prod_{j=1}^{n}p_{j}
\end{eqnarray*}
}

\noindent \emph{and so}

\emph{ 
\[
k^{n-1}=\lambda^{n}\prod_{j=1}^{n}p_{j}
\]
}

\noindent \emph{Thus, $\lambda=k^{\frac{n-1}{n}}\prod_{j=1}^{n}p_{j}^{-\frac{1}{n}}$.
Dividing the last of the Lagrange multiplier (the constraint equation)
by the $j^{th}$ equation yields 
\[
x_{j}=k\lambda^{-1}p_{j}^{-1}
\]
}

\noindent \emph{Since the impermanent loss formula involves $p_{j}x_{j}$
expressions, it is convenient to note that}

\noindent \emph{ 
\begin{eqnarray*}
p_{j}x_{j} & = & k\lambda^{-1}\\
 & = & k^{\frac{1}{n}}\prod_{l=1}^{n}p_{l}^{\frac{1}{n}}
\end{eqnarray*}
}

\noindent \emph{This demonstrates that the stable state for a constant
product market maker occurs when each collection of tokens in the
AMM has the same value. Substituting this back into the original equation,
we get}

\noindent \emph{ 
\begin{eqnarray*}
IL & = & \frac{\sum_{j=1}^{n}p_{j}^{f}x_{j}^{f}}{\sum_{j=1}^{n}p_{j}^{f}x_{j}^{i}}-1\\
 & = & \frac{n\left(\prod_{j=1}^{n}p_{j}^{f}\right)^{\frac{1}{n}}}{\left(\prod_{j=1}^{n}p_{j}^{i}\right)^{\frac{1}{n}}\sum_{j=1}^{n}\frac{p_{j}^{f}}{p_{j}^{i}}}-1\\
 & = & \frac{\left(\prod_{j=1}^{n}\frac{p_{j}^{f}}{p_{j}^{i}}\right)^{\frac{1}{n}}}{\frac{1}{n}\sum_{j=1}^{n}\frac{p_{j}^{f}}{p_{j}^{i}}}-1
\end{eqnarray*}
}

\noindent \emph{We now introduce the analogs to the exchange rates
in Example \ref{Ex 3} by defining $m_{j}=\frac{p_{j}}{p_{1}}$, so
that $m_{j}$ is the exchange rate between tokens $1$ and $j$. Then,
we see that}

\noindent \emph{ 
\begin{eqnarray*}
IL & = & \frac{\left(\prod_{j=1}^{n}\frac{p_{j}^{f}}{p_{j}^{i}}\frac{p_{1}^{i}}{p_{1}^{f}}\right)^{\frac{1}{n}}}{\frac{1}{n}\sum_{j=1}^{n}\frac{p_{j}^{f}}{p_{j}^{i}}\frac{p_{1}^{i}}{p_{1}^{f}}}-1\\
 & = & \frac{\left(\prod_{j=1}^{n}\frac{m_{j}^{f}}{m_{j}^{i}}\right)^{\frac{1}{n}}}{\frac{1}{n}\sum_{j=1}^{n}\frac{m_{j}^{f}}{m_{j}^{i}}}-1\\
 & = & \frac{\left(\prod_{j=1}^{n}t_{j}\right)^{\frac{1}{n}}}{\frac{1}{n}\sum_{j=1}^{n}t_{j}}-1
\end{eqnarray*}
}

\noindent \emph{Thus, noting that $t_{1}=1$ yields the final result:}

\emph{ 
\[
IL=\frac{\text{Geometric Mean of the\thinspace\thinspace}t_{j}\text{'s}}{\text{Arithmetic Mean of the\thinspace\thinspace}t_{j}\text{'s}}-1
\]
} 
\end{proof}
Similarly to Example \ref{Ex 3}, $t_{j}$ is the quotient of the
final exchange rate and initial exchange rate between assets $j$
and $1$. The fact that the $IL\leq0$ for constant product market
makers is a direct consequence of the arithmetic mean-geometric mean
inequality applied to these exchange rate quotients. Again, the key
takeaway is that impermanent loss is best understood in terms of these
$n-1$ $t_{j}$ parameters, not the $2n$ price parameters. This result
has important economic implications in addition to its overall simplifying
nature.

The parameter reduction in the impermanent loss formula from $n$
initial prices and $n$ final prices to $n-1$ quotients of exchange
rates can be decomposed into two parts. The first reduction consists
of using relative pricing of one asset against another as opposed
to the raw prices (in our case, we have been pricing assets in terms
of the first asset). This is typically referred to as \emph{numeraire}
independence in finance, but we will refer to it as \emph{price level
independence}. The next reduction takes us from separate relative
initial and final exchange rates to quotients of final and initial
exchange rates. We call this \emph{exchange rate level independence},
since the absolute levels of the initial and final exchange rates
do not matter, only the quotients. In the next section, we will show
that impermanent loss for all AMMs exhibits \emph{price level independence},
but not necessarily \emph{exchange rate level independence}.

\section{Price Level Independence}

The goal of this section is to generalize the above idea of reducing
the number of parameters involved in defining impermanent loss. We
will begin by showing that impermanent loss for all AMMs is price
level independent. Price level independence is intuitive and not particularly
deep in its own right, but we wish to build towards understanding
exchange rate level independence systematically. To this end, we need
the following lemma, which formalizes the concept introduced in the
background section that rescaling price vectors does not affect stable
points. In this section, we think of there being a fixed measuring
currency through which $p$ and hence $V(x)$ are defined. In other
words, $2p$ corresponds to a world in which all prices have doubled,
not one in which the measuring currency has changed. 
\begin{lem}
Let $A$ be an AMM and for fixed $k$, consider the liquidity surface
$A=k$. We know that we can express the stable point on the $A=k$
surface as a function of price $p$: $x=x(p)$. The function $x(p)$
is homogeneous of degree $0$ in $p$. In other words,

\[
x(cp)=x(p)
\]

\noindent for all $c>0$, and so stable points are price level independent. 
\end{lem}

\begin{proof}
\emph{Fix a price vector $p$ and $c>0$. Let $x^{*}=x(p)$. Then,
$x^{*}$ satisfies the following Lagrange multiplier equations:}

\emph{ 
\begin{eqnarray*}
\nabla_{x}A(x^{*}) & = & \lambda p\\
A(x^{*}) & = & k
\end{eqnarray*}
}

\noindent \emph{for some $\lambda$. It follows that $x^{*}$ satisfies}

\emph{ 
\begin{eqnarray*}
\nabla_{x}A(x^{*}) & = & \tilde{\lambda}cp\\
A(x^{*}) & = & k
\end{eqnarray*}
}

\noindent \emph{for $\tilde{\lambda}=\frac{\lambda}{c}$. By uniqueness
of stable points, $x(cp)=x^{*}$.} 
\end{proof}
Price level independence for the impermanent loss of an AMM follows
immediately. 
\begin{thm}
\label{Thm 6}The impermanent loss of an AMM exhibits price level
independence. In other words, impermanent loss can be expressed entirely
in terms of exchange rates of tokens relative to the first token (or
any token for that matter). 
\end{thm}

\begin{proof}
\,

\noindent 
\begin{eqnarray*}
IL & = & \frac{\sum_{j=1}^{n}p_{j}^{f}x_{j}(p^{f})}{\sum_{j=1}^{n}p_{j}^{f}x_{j}(p^{i})}-1\\
 & = & \frac{\sum_{j=1}^{n}\frac{p_{j}^{f}}{p_{1}^{f}}x_{j}(p^{f})}{\sum_{j=1}^{n}\frac{p_{j}^{f}}{p_{1}^{f}}x_{j}(p^{i})}-1=\frac{\sum_{j=1}^{n}\frac{p_{j}^{f}}{p_{1}^{f}}x_{j}\left(\frac{p^{f}}{p_{1}^{f}}\right)}{\sum_{j=1}^{n}\frac{p_{j}^{f}}{p_{1}^{f}}x_{j}\left(\frac{p^{i}}{p_{1}^{i}}\right)}-1
\end{eqnarray*}

\noindent \emph{where the last equality follows from the lemma. Using
our previous notation and defining $m_{j}=\frac{p_{j}}{p_{1}}$, we
obtain}

\emph{ 
\[
IL=\frac{\sum_{j=1}^{n}m_{j}^{f}x_{j}(m^{f})}{\sum_{j=1}^{n}m_{j}^{f}x_{j}(m^{i})}-1
\]
} 
\end{proof}
Indeed, it is clear from the proof that we could express impermanent
loss using exchange rates relative to any fixed asset. Along these
lines, it will be convenient at times to change the base token from
the first token to another of the $n$ tokens. In other words, we
might define 
\[
m_{j}=m_{j,k}=\frac{p_{j}}{p_{k}}
\]
so that we use the $k^{\text{th}}$ token as the base token when computing
exchange rates. In this case, $m_{k}=1$. To keep the notation uncluttered,
though, we will suppress the $k$ and make sure that it is clear from
context when important. 
\begin{example}
\emph{\label{Ex 7}We will give context for the price level independence
of impermanent loss through an example involving a two dimensional
AMM $A(x_{1},x_{2})$. We make no assumptions about the specific formula
for this AMM and we investigate a fixed liquidity curve $A(x_{1},x_{2})=k$.
The price vectors $p^{i}$ and $p^{f}$ are perpendicular to the liquidity
curve at the initial and final stable states $x^{i}$ and $x^{f}$
respectively. The above theorem formalizes the idea that impermanent
loss can be understood using slopes alone. In other words, the sizes
of the price vectors $p^{i}$ and $p^{f}$ are irrelevant. If we use
the second token as the base token, we see that the normalized versions
of $p^{i}$ and $p^{f}$ are $\left(\frac{p_{1}^{i}}{p_{2}^{i}},1\right)$
and $\left(\frac{p_{1}^{f}}{p_{2}^{f}},1\right)$ respectively. Thus,
in two dimensions, each price vector is associated to a slope as is
pictured, and these slopes are the negatives of the exchange rates.
Instead of the four parameters $p_{1}^{i}$, $p_{2}^{i}$, $p_{1}^{f}$,
and $p_{2}^{f}$, we only need two parameters: $m_{1}^{i}$ and $m_{1}^{f}$.
Refer to Figure \ref{Fig 4} for visualization.} 
\end{example}

\begin{figure}[t]
\noindent \begin{centering}
\includegraphics[scale=0.75]{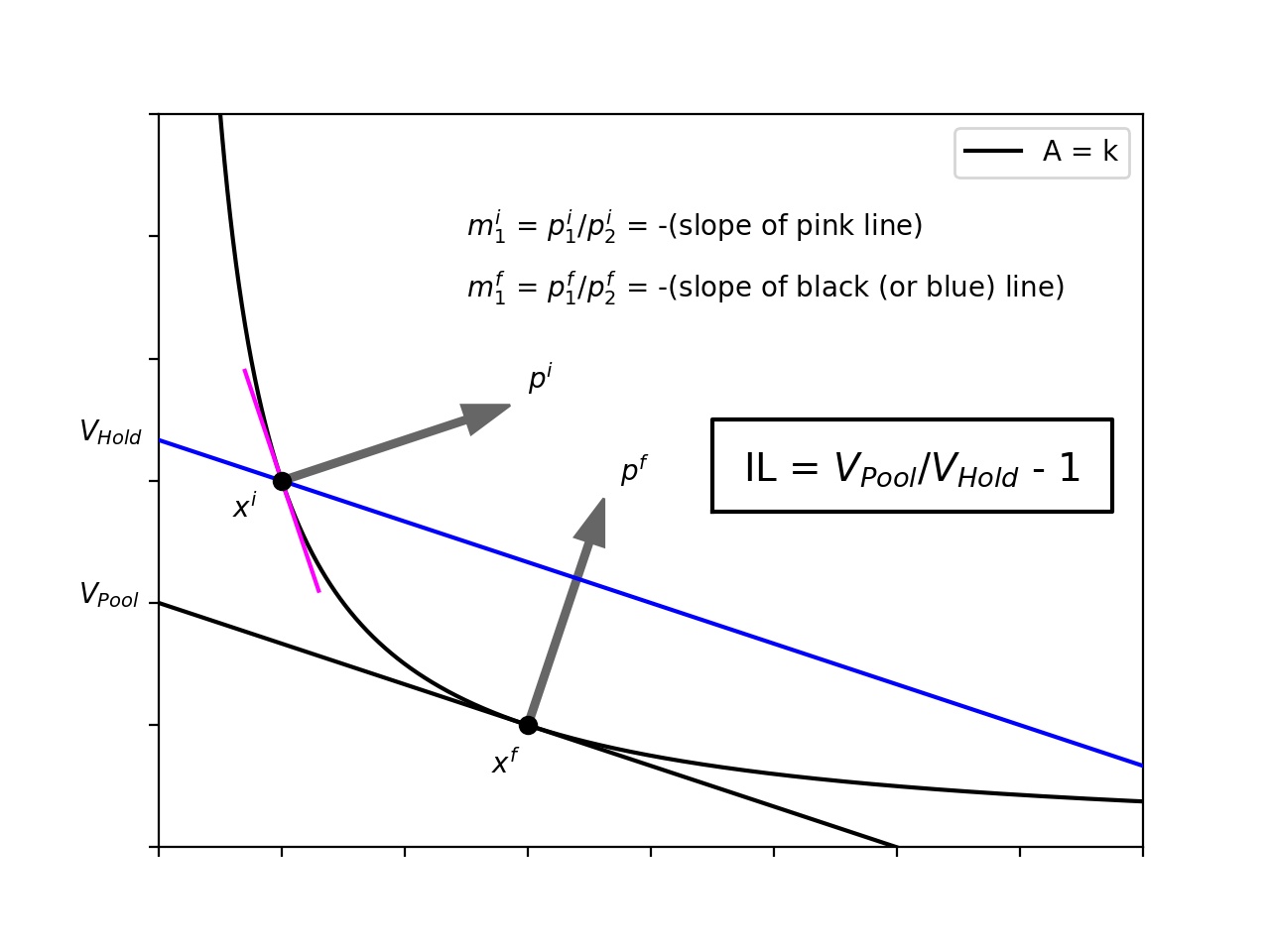} 
\par\end{centering}
\caption{\label{Fig 4}The tangent lines to the liquidity curve at $x^{i}$
and $x^{f}$ are the negatives of these exchange rates. The blue and
black lines are level curves of the value function $V(x)=p^{f}\cdot x$.
It is convenient to use the $x_{2}$ intercepts of these level curves
to quantify value using the $x_{2}$ currency.}
\end{figure}

Another way of interpreting Theorem \ref{Thm 6} is that impermanent
loss is agnostic to the measuring currency used for both $p^{i}$
and $p^{f}$. Indeed, we could use a different measuring currency
for $p^{i}$ and $p^{f}$ without affecting impermanent loss, though
we will not do this. The above figure offers some insight into how
we can choose the measuring currency so that $V(x)=p\cdot x$ can
be gleaned geometrically. More concretely, if the $x_{2}$ token is
the measuring currency, the pool value is the $x_{2}$ intercept\footnote{We could use any of the tokens as the measuring currency, and then
the pool value would be the corresponding intercept value.} of the line $L$ with the following properties: 
\begin{itemize}
\item $L$ passes through the current pool state 
\item $L$ is perpendicular to $p$ 
\end{itemize}
The slope of $L$ is the negative of the exchange rate between the
two tokens. This line $L$ will be tangent to the $A=k$ curve if
the pool state is the current stable state. In the above example,
if the individual had just held their assets, the value of their assets
would have been equivalent to the value of $V_{\text{Hold}}$ $x_{2}$
tokens. By providing liquidity, the value of their assets is instead
$V_{\text{Pool}}$ $x_{2}$ tokens. This diagram makes it clear why
the convexity condition for AMMs ensures that impermanent loss will
always be negative.

\section{The Legendre Transform and the Connection Between Value and State}

\label{TheLegend}

As introduced in \cite{AEC21}, there is an elegant mathematical relationship
between automated market makers and their corresponding portfolio
value functions. In that paper, the relationship is formulated between
the AMM $A(x)$ and its value function $V(x)$ via the Legendre-Fenchel
transform. Example \ref{Ex 7} motivates us to proceed a bit differently.
The core object that we will work with instead of $A(x)$ is $f(x_{1},x_{2},\ldots x_{n-1})$,
which we will define as the function that defines the specific level
curve (or surface) $A(x)=k$: 
\[
A(x_{1},\ldots,x_{n-1},f(x_{1},\ldots,x_{n-1}))=k
\]

\noindent For those looking for a visual encapsulation of this section,
and possibly to skip some of the rigmarole, the graphs at the end
of this section illustrate the main ideas.

By our AMM assumptions and the implicit function theorem, $x_{n}=f(x_{1},x_{2},\ldots,x_{n-1})$
is a continuously differentiable surface. Henceforth, we will use
the $x_{n}$ token as the measuring currency through which $p$ and
hence $V(x)$ are defined. In other words, $p_{n}=1$ and $p_{i}$
is the exchange rate between token $i$ and token $n$ for $1\leq i\leq n-1$.
To maintain consistency with the previous section, we will use $m_{i}:=p_{i}$
to reference these exchange rates. Define $\bar{V}$ as

\noindent 
\begin{eqnarray*}
V(x_{1},\ldots,x_{n-1}) & := & V(x_{1},\ldots,x_{n-1},f(x_{1},\ldots,x_{n-1}))\\
 & = & f(x_{1},\ldots,x_{n-1})+\sum_{j=1}^{n-1}m_{j}x_{j}
\end{eqnarray*}

\noindent and define

\noindent 
\begin{eqnarray*}
W(m_{1},\ldots,m_{n-1}) & := & V(x_{1}(m_{1},\ldots,m_{n-1},1),\ldots,x_{n-1}(m_{1},\ldots,m_{n-1},1))
\end{eqnarray*}

In other words, given the exchange rates $m_{1},m_{2},\ldots,m_{n-1}$,
a unique stable point $x$ is determined, and $W$ is the value of
the corresponding stable portfolio. We will use $x_{i}(m_{1},\ldots,m_{n-1},1)$
interchangeably with $x_{i}(m_{1},\ldots,m_{n-1})$, mainly to simplify
notation and to avoid introducing new notation. From Example \ref{Ex 7},
we see that $-m_{1}=f'(x_{1}(m_{1}))$. In other words, at stable
points, we can obtain the exchange rate as the negative derivative
of $f$ \footnote{In Example \ref{Ex 7}, $x^{f}$ is the stable point at the final
exchange rate $m_{1}^{f}$, while $x^{i}$ is not a stable point.}. The following lemma proves and generalizes this observation to cases
in which $x$ is a function of $n-1$ exchange rates. 
\begin{lem}
\label{Lem 8}With $f(x_{1},x_{2},\ldots,x_{n-1})$ and $m_{i}$ for
$1\leq i\leq n-1$ defined as above, we have that 
\[
-m_{i}=f_{x_{i}}(x_{1}(m_{1},\ldots,m_{n-1}),\ldots,x_{n-1}(m_{1},\ldots,m_{n-1}))
\]
\end{lem}

\begin{proof}
\emph{For the stable point $x(m)$, we have the Lagrange multiplier
equations}

\noindent \emph{ 
\begin{eqnarray*}
\nabla_{x}A(x(m)) & = & \lambda m\\
A(x(m)) & = & k
\end{eqnarray*}
}

\noindent \emph{for some $\lambda$. Differentiating the relation
\[
A(x_{1},\ldots,x_{n-1},f(x_{1},\ldots,x_{n-1}))=k
\]
with respect to $x_{i}$ yields 
\[
A_{x_{i}}(x_{1},\ldots,x_{n-1},f(x_{1},\ldots,x_{n-1}))+A_{x_{n}}(x_{1},\ldots,x_{n-1},f(x_{1},\ldots,x_{n-1}))f_{x_{i}}(x_{1},\ldots,x_{n-1})=0
\]
\[
\implies f_{x_{i}}(x_{1},\ldots,x_{n-1})=-\frac{A_{x_{i}}(x_{1},\ldots,x_{n-1},f(x_{1},\ldots,x_{n-1}))}{A_{x_{n}}(x_{1},\ldots,x_{n-1},f(x_{1},\ldots,x_{n-1}))}
\]
Substituting $x_{i}(m_{1},\ldots,m_{n-1})$ into the above yields}

\emph{ 
\[
f_{x_{i}}(x_{1}(m_{1},\ldots,m_{n-1}),\ldots,x_{n-1}(m_{1},\ldots,m_{n-1}))=-\frac{\lambda m_{i}}{\lambda}=-m_{i}
\]
}

\noindent \emph{and so the result is proved.} 
\end{proof}
\noindent The next lemma frames the same relationship from the viewpoint
that the $m_{i}$ are functions of the $x_{i}$. 
\begin{lem}
\label{Lem 9}With $f(x_{1},x_{2},\ldots,x_{n-1})$ defined as above
and the $m_{i}(x_{1},\ldots,x_{n-1})$ functions for $1\leq i\leq n-1$
defined as the natural inverses to the $x_{i}(m_{1},\ldots,m_{n-1})$
functions, we have that 
\[
m_{i}(x_{1},\ldots,x_{n-1})=-f_{x_{i}}(x_{1},\ldots,x_{n-1})
\]
\end{lem}

\begin{proof}
\emph{From Lemma \ref{Lem 8}, it is clear that the map 
\[
(m_{1},\ldots,m_{n-1})\mapsto(x_{1}(m_{1},\ldots,m_{n-1}),\ldots,x_{n-1}(m_{1},\ldots,m_{n-1}))
\]
is injective. Thus, $m_{i}(x_{1},\ldots,x_{n-1})$ for $1\leq i\leq n-1$
is well defined and plugging these expressions into both sides of
the equality in Lemma \ref{Lem 8} yields the result. } 
\end{proof}
The following theorem gives a precise formulation of how $A(x_{1},\ldots,x_{n-1})$
and $W(m_{1},\ldots,m_{n-1})$ are dual functions. For readers unfamiliar
with the Legendre transform, the remainder of the paper is self-contained:
all properties that we need are readily derived from the formulas
that we have established. As such, the following theorem is just a
form of mental bookkeeping. We recommend \cite{ZRM09} as a useful
resource to learn more about the Legendre transform and its applications. 
\begin{thm}
\label{Thm 10}Denote the Legendre transform of a function $g(x_{1},\ldots,x_{n-1})$
in all of its variables as $\mathcal{L}(g)$. In other words, 
\[
\mathcal{L}(g)(m_{1},\ldots,m_{n-1}):=-g(x_{1}(m_{1},\ldots,m_{n-1}),\ldots,x_{n-1}(m_{1},\ldots,m_{n-1}))+\sum_{i=1}^{n-1}m_{i}x_{i}(m_{1},\ldots,m_{n-1})
\]
where $m_{i}$ is defined as 
\[
m_{i}(x_{1},\ldots,x_{n-1}):=-g_{x_{i}}(x_{1},\ldots,x_{n-1})
\]
We then have that 
\[
\mathcal{L}(f)=W
\]
\end{thm}

\begin{proof}
\emph{The result follows from Lemma \ref{Lem 8}.} 
\end{proof}
Again, the relationship between $A$ and $V$ via the Legendre-Fenchel
transform was pointed out in \cite{AEC21}, and this observation accounts
for the heavy lifting. It can be cumbersome in some settings, however,
to think in terms of $A$ (a family of level curves or surfaces) as
opposed to the single AMM surface given by $x_{n}=f(x_{1},\ldots,x_{n-1})$.
For instance, to replicate a payoff function $W(x_{1},\ldots,x_{n-1})$,
one can extend $W$ to $V$ via 1-homogeneity, transform $V$ via
the Legendre-Fenchel transform into an $\tilde{A}$, and then manipulate
$\tilde{A}$ to obtain a more amenable AMM form $A$. Alternatively,
one can use Theorem \ref{Thm 10} to pass from $W$ to $f$ directly.
The graphs in Figure \ref{Fig 5} illustrate the relationship between
$f$ and $W$ in two dimensions.

\begin{figure}[t]
\begin{centering}
\begin{tabular}{cc}
\includegraphics[scale=0.5]{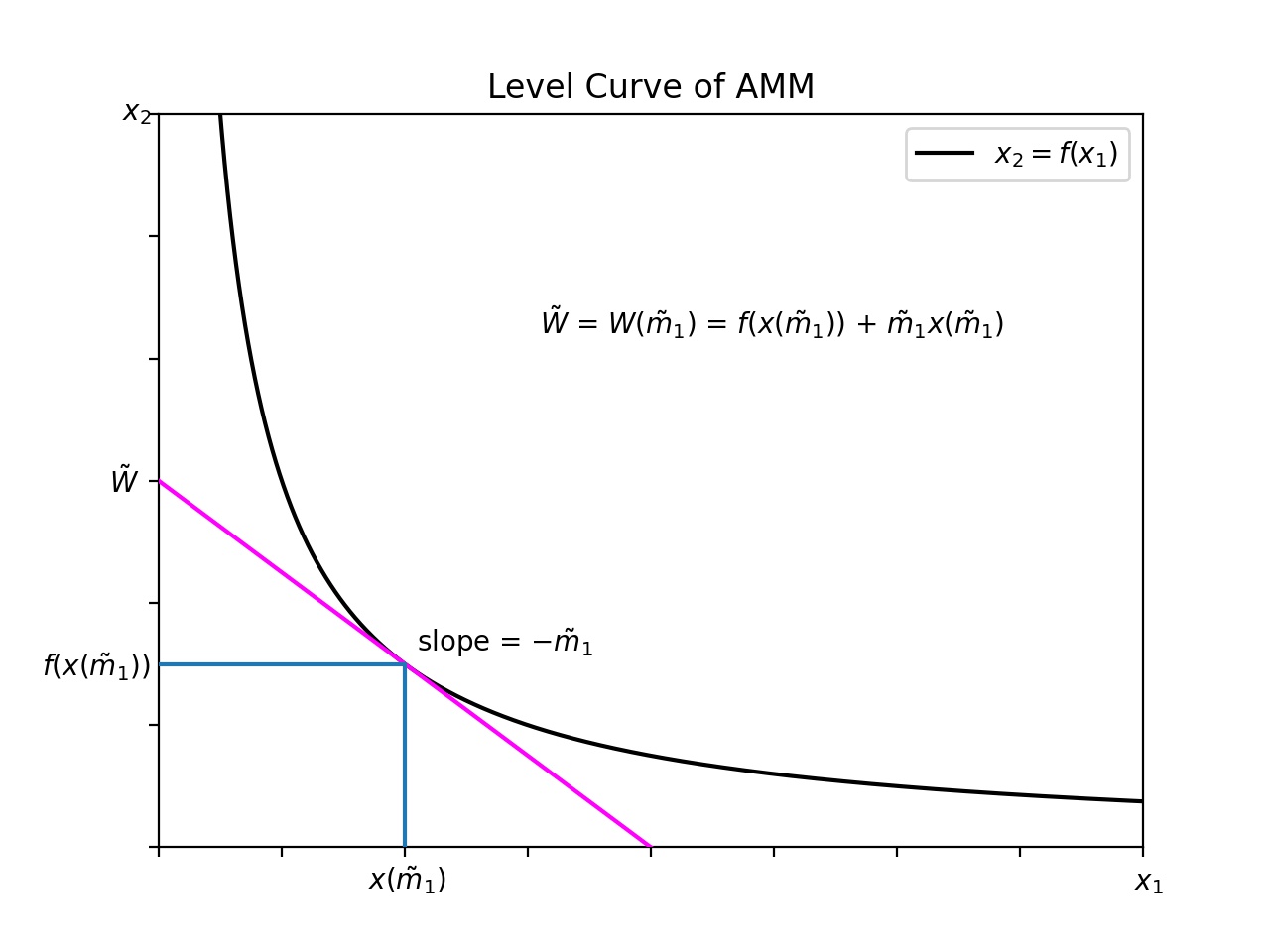}  & \includegraphics[scale=0.5]{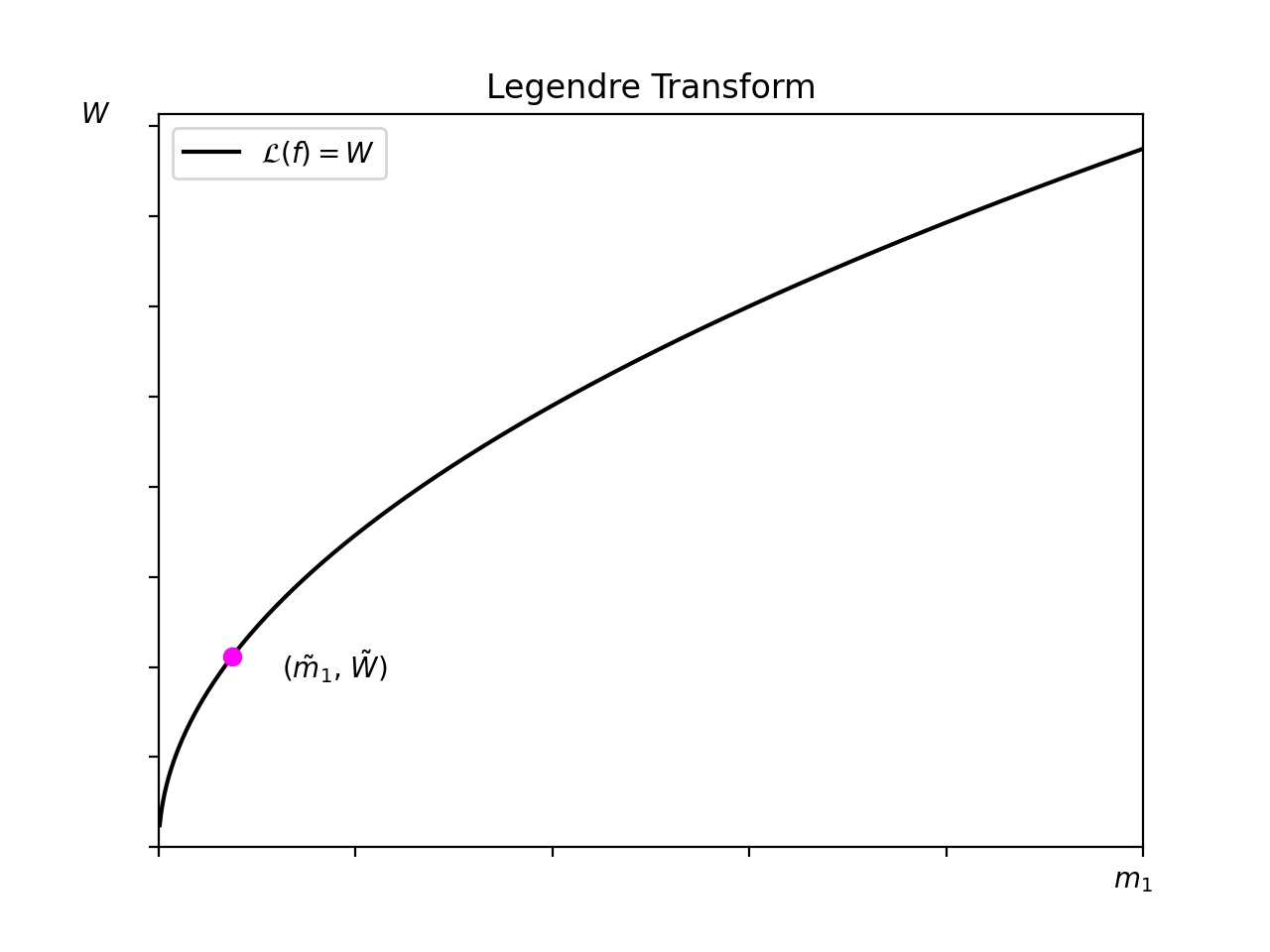}\tabularnewline
\end{tabular}
\par\end{centering}
\caption{\label{Fig 5}A level curve of an AMM $A(x_{1},x_{2})=k$ is displayed,
with $f(x_{1})$ satisfying $A(x_{1},f(x_{1}))=k$. For a given exchange
rate $m_{1}$, there is an easy way to see what the value of the arbitrageur
driven portfolio will be. This value is $\tilde{W}$ and it is measured
using $x_{2}$ tokens. The $(\tilde{m}_{1},\tilde{W})$ pair is plotted
in the bottom graph. By convexity, for every $m_{1}$, there will
be a unique $W$, and stitching these points together yields the Legendre
transform pictured in the graph on the right.}
\end{figure}

\section{Exchange Rate Level Independence}

In general, price level independence of impermanent loss allows us
to reduce the number of parameters from $2n$ to $2(n-1)$. When we
derived the impermanent loss formula for the $n$-dimensional constant
product market maker, however, we were able to reduce the number of
relevant parameters even further to $n-1$. These $n-1$ parameters
were quotients of the initial and final exchange rates (relative to
a fixed token). 
\begin{defn}
\emph{Let $A$ be an AMM and fix a liquidity surface $A=k$. Let $p^{i},p^{f}$
be price vectors and denote the corresponding stable states $x^{i},x^{f}$.
Define $m_{j}=\frac{p_{j}}{p_{n}}$ and $t_{j}=\frac{m_{j}^{f}}{m_{j}^{i}}$.
We say that the AMM is exchange rate level independent (i.e. ERLI)
if $IL$ can be rewritten as a function of the exchange rate quotients
$t_{1},...,t_{n}$ (where $t_{n}$ can be omitted because it is always
equal to $1$).} 
\end{defn}

\noindent In other words, the AMM satisfies the ERLI conditions if
its impermanent loss can be written purely as a function of exchange
rate quotients. Before proceeding, let us consider an example that
does not meet this requirement.
\begin{example}
\emph{The Curve StableSwap AMM introduced in \cite{Ego19} does not
satisfy the ERLI condition. We will demonstrate this by using their
proposed two dimensional AMM with parameter $A=1$ (not to be confused
with the $A$ that we have been using to denote our AMM). The $D=1$
level curve of this AMM (where we use the notation from the paper)
is given by 
\[
4(x_{1}+x_{2})+1=4+\frac{1}{4x_{1}x_{2}}
\]
which can be simplified to 
\[
16x_{1}^{2}x_{2}+16x_{1}x_{2}^{2}-12x_{1}x_{2}=1
\]
Note that 
\[
x_{2}=\frac{12x_{1}-16x_{1}^{2}+\sqrt{256x_{1}^{4}-384x_{1}^{3}+144x_{1}^{2}+64x_{1}}}{32x_{1}}
\]
is not homogeneous in $x_{1}$. There is a better way to see this
without solving for $x_{2}$. Assume for a contradiction that $x_{2}=cx_{1}^{\alpha}$
for some $c$. Then, 
\[
c_{1}x_{1}^{\alpha+2}+c_{2}x_{1}^{2\alpha+1}+c_{3}x^{\alpha+1}+c_{4}=0
\]
for all $x_{1}$ and some nonzero constants $c_{1},c_{2},c_{3},$
and $c_{4}$. This is impossible since the $x_{1}^{\max\{\alpha+2,2\alpha+1,0\}}$
term will grow unchecked in $x_{1}$. In this section, we will show
that the lack of homogeneity of $x_{2}=f(x_{1})$ implies that the
impermanent loss function for this liquidity curve does not satisfy
the ERLI condition.} 
\end{example}

Our first goal will be to write impermanent loss in terms of the value
function $W$ defined in the previous section. To this end, we need
the following lemma. We warn the reader of the notational simplification

\noindent 
\begin{eqnarray*}
m & = & (m_{1},\ldots,m_{n-1},1)=(m_{1},\ldots,m_{n-1})\\
x & = & (x_{1},\ldots,x_{n-1},x_{n})=(x_{1},\ldots,x_{n-1})
\end{eqnarray*}

\noindent and define $f$ and $W$ as in the previous section. 
\begin{lem}
\label{Lem 13}For $1\leq j\leq n-1$, we have

\[
x_{j}(m)=W_{m_{j}}(m)
\]
\end{lem}

\begin{proof}
\emph{We start with relation 
\[
W(m)-f(x_{1}(m),\ldots,x_{n-1}(m))=\sum_{i=1}^{n-1}m_{i}x_{i}(m)
\]
and differentiate both sides of the equality with respect to $m_{j}$
to obtain}

\emph{ 
\[
W_{m_{j}}(m)-\sum_{i=1}^{n-1}f_{x_{i}}(x_{1}(m),\ldots,x_{n-1}(m))[x_{i}]_{m_{j}}(m)=x_{j}(m)+\sum_{i=1}^{n-1}m_{i}[x_{i}]_{m_{j}}(m)
\]
}

\noindent \emph{which yields the result since 
\[
-f_{x_{i}}(x_{1}(m),\ldots,x_{n-1}(m))=m_{i}
\]
for $1\leq i\leq n-1$.} 
\end{proof}
We are now ready to recast impermanent loss in terms of $W$. 
\begin{thm}
\label{Thm 14}The expression for impermanent loss in terms of $W$
is 
\[
IL=\frac{W(m^{f})}{P(m^{f})}-1
\]
where $P$ is the linear approximation to $W$ at $m^{i}$. 
\end{thm}

\begin{proof}
\emph{From Theorem \ref{Thm 6}, we have}

\noindent \emph{ 
\begin{eqnarray*}
IL & = & \frac{\sum_{j=1}^{n}m_{j}^{f}x_{j}(m^{f})}{\sum_{j=1}^{n}m_{j}^{f}x_{j}(m^{i})}-1\\
 & = & \frac{x_{n}(m^{f})+\sum_{j=1}^{n-1}m_{j}^{f}x_{j}(m^{f})}{\sum_{j=1}^{n}(m_{j}^{f}-m_{j}^{i})x_{j}(m^{i})+\sum_{j=1}^{n}m_{j}^{i}x_{j}(m^{i})}-1\\
 & = & \frac{f(x_{1}(m^{f}),\ldots,x_{n-1}(m^{f}))+\sum_{j=1}^{n-1}m_{j}^{f}x_{j}(m^{f})}{\left[\sum_{j=1}^{n-1}(m_{j}^{f}-m_{j}^{i})x_{j}(m^{i})\right]+f(x_{1}(m^{i}),\ldots,x_{n-1}(m^{i}))+\sum_{j=1}^{n-1}m_{j}^{i}x_{j}(m^{i})}-1\\
 & = & \frac{W(m^{f})}{W(m^{i})+\sum_{j=1}^{n-1}(m_{j}^{f}-m_{j}^{i})x_{j}(m^{i})}-1\\
 & = & \frac{W(m^{f})}{W(m^{i})+\sum_{j=1}^{n-1}(m_{j}^{f}-m_{j}^{i})W_{m_{j}}(m^{i})}-1
\end{eqnarray*}
}

\noindent \emph{where we have used Lemma \ref{Lem 13} to obtain the
final equality.} 
\end{proof}
We will finish the section by providing a necessary and sufficient
condition on $A$ for its impermanent loss to satisfy the ERLI condition.
To this end, we need a lemma that connects the homogeneity properties
of $f$ to the homogeneity properties of $W$. 
\begin{lem}
\label{Lem 15}The following two statements are equivalent. 
\end{lem}

\begin{enumerate}
\item \emph{For each $1\leq j\leq n-1$, $f(x_{1},\ldots,x_{n-1})$ is homogeneous
of degree $\alpha_{j}$ in $x_{j}$.} 
\item \emph{For each $1\leq j\leq n-1$, $W(m_{1},\ldots,m_{n-1})$ is homogeneous
of degree $\beta_{j}$ in $m_{j}$.} 
\end{enumerate}
\emph{Furthermore, }

\noindent \emph{ 
\[
\beta_{j}=\frac{\alpha_{j}}{-1+\sum_{i=1}^{n-1}\alpha_{i}}
\]
Note that the case that $\sum_{i=1}^{n-1}\alpha_{i}=1$ is precluded
by the condition that the AMM is strictly convex}\footnote{\emph{To see this, consider the straight line parametric path $r(t)$
in $\mathbb{R}^{n-1}$ given by $x_{i}(t)=t$, for $1\leq i\leq n-1$.
If $\sum_{i=1}^{n-1}\alpha_{i}=1$, then $f(r(t))=ct$ for some constant
$c$. Fix $0<\lambda<1$. For any points $p=(a,\ldots,a)$ and $q=(b,\ldots,b)$
falling on $r(t)$, $\lambda f(p)+(1-\lambda)f(q)=\lambda ca+(1-\lambda)cb=c(\lambda a+(1-\lambda)b)=f(\lambda p+(1-\lambda)q)$. }}\emph{.} 
\begin{proof}
\emph{From the way that $x(m)$ was originally defined, it is clear
that $W(m)$ is the solution to the following value minimization problem
\[
W(m)=\min_{x\in\mathbb{R}^{n-1}}\left[m\cdot x+f(x)\right]
\]
Assume that the first statement holds. Then, with $\beta_{j}$ defined
as in the statement of the theorem,}

\emph{ 
\begin{eqnarray*}
W(m_{1},\ldots,cm_{j},\ldots,m_{n-1}) & = & \min_{x\in\mathbb{R}^{n-1}}\left[cm_{j}x_{j}+\sum_{i\neq j}m_{i}x_{i}+f(x_{1},\ldots,x_{j},\ldots,x_{n-1})\right]\\
 & = & \min_{x\in\mathbb{R}^{n-1}}\left[\sum_{i=1}^{n-1}m_{i}x_{i}+f(x_{1},\ldots,c^{-1}x_{j},\ldots,x_{n-1})\right]\\
 & = & \min_{x\in\mathbb{R}^{n-1}}\left[\sum_{i=1}^{n-1}m_{i}x_{i}+c^{-\alpha_{j}}f(x_{1},\ldots,x_{j},\ldots,x_{n-1})\right]\\
 & = & \min_{x\in\mathbb{R}^{n-1}}\left[\sum_{i=1}^{n-1}m_{i}x_{i}+c^{\beta_{j}}f(c^{-\beta_{j}}x_{1},\ldots,c^{-\beta_{j}}x_{j},\ldots,c^{-\beta_{j}}x_{n-1})\right]\\
 & = & \min_{x\in\mathbb{R}^{n-1}}\left[\sum_{i=1}^{n-1}m_{i}c^{\beta_{j}}x_{i}+c^{\beta_{j}}f(x_{1},\ldots,x_{j},\ldots,x_{n-1})\right]\\
 & = & c^{\beta_{j}}W(m_{1},\ldots,m_{j},\ldots,m_{n-1})
\end{eqnarray*}
}

\noindent \emph{Thus, we have shown that statement $1$ implies statement
$2$. For the other direction, we note that }

\emph{ 
\begin{equation}
f(x)=\max_{m\in\mathbb{R}^{n-1}}\left[V(m)-m\cdot x\right]\label{MaxDuality}
\end{equation}
}

\noindent \emph{To see this, we can unpack the meanings of $f(x)$
and $V(m)-m\cdot x$. Fix $x=(x_{1},\ldots,x_{n-1})$ as the quantities
for tokens $1$ through $n-1$ in the AMM. Then: $V(m)$ is the overall
value in the AMM (measured in $x_{n}$ tokens) at the stable point
corresponding to exchange rate vector $m$, and $m\cdot x$ is the
value of tokens tokens $1$ through $n-1$ (measured in $x_{n}$ tokens)
assuming exchange rate vector $m$. Thus, $V(m)-m\cdot x$ is maximized
by the exchange vector $m$ that allows for the most $x_{n}$ tokens.
This exchange rate vector $m$ corresponds to the world with the most
profitable ``hold'' strategy. There are other more technical ways
to show that above equation holds, but we omit these for expository
flow. Using equation (\ref{MaxDuality}), we can prove the other direction
of the lemma in an analogous manner to the first direction.} 
\end{proof}
We need two final lemmas relating the homogeneity of $A$ to the homogeneity
of $f$ before proceeding with the final theorem of this section. 
\begin{lem}
\label{Lem 16}If $A$ is homogeneous in each of its $n$ coordinates,
then $f$ is homogeneous in each of its $n-1$ inputs. 
\end{lem}

\begin{proof}
\emph{Assume that $A$ is homogeneous of degree $\gamma_{l}$ in $x_{l}$
for all $1\leq l\leq n$. Then,}

\emph{ 
\begin{eqnarray*}
k & = & A(x_{1},\ldots,x_{n})\\
 & = & A(1,\ldots,1)x_{1}^{\gamma_{1}}\ldots x_{n}^{\gamma_{n}}
\end{eqnarray*}
}

\noindent \emph{implying by uniqueness that 
\[
x_{n}=f(x_{1},\ldots,x_{n-1})=\left[\frac{k}{A(1,\ldots,1)}\right]^{-\frac{1}{\gamma_{n}}}x_{1}^{-\frac{\gamma_{1}}{\gamma_{n}}}\ldots x_{n-1}^{-\frac{\gamma_{n-1}}{\gamma_{n}}}
\]
} 
\end{proof}
\begin{lem}
\label{Lem 17}Assume for each $k$, the function $x_{n}=f(x_{1},\ldots,x_{n-1})$
induced by $A=k$ is homogeneous in each of its $n-1$ coordinates\footnote{Technically, $f$ should be subscripted by $k$, but we drop the $k$
to simplify notation.}. Then, there is a function $B(x_{1},\ldots,x_{n})$ that is homogeneous
in each of its coordinates and such that $A(x)=g(B(x))$, where $g$
is a smooth function of one variable. 
\end{lem}

\begin{proof}
\emph{For a fixed $k$, assume that $f$ is homogeneous of degree
$\lambda_{l}$ for $1\leq l\leq n-1$. We start with 
\[
x_{n}=f(x_{1},\ldots,x_{n-1})=Cx_{1}^{\lambda_{1}}\ldots x_{n-1}^{\lambda_{n-1}}
\]
and proceed to obtain 
\[
\tilde{C}=x_{1}^{\gamma_{1}}\ldots x_{n-1}^{\gamma_{n-1}}x_{n}
\]
where $\gamma_{j}=-\lambda_{j}$ for $1\leq j\leq n-1$. Observe that
all level surfaces $A=k$ can be expressed using the same $\gamma$
exponents since, otherwise, the level surfaces would cross. In other
words, setting two surface equations like the one above with different
$\tilde{C}$'s equal to one another would lead to solutions:}

\begin{eqnarray*}
Cx_{1}^{\lambda_{1}}\ldots x_{n-1}^{\lambda_{n-1}} & = & \tilde{C}x_{1}^{\tilde{\lambda}_{1}}\ldots x_{n-1}^{\tilde{\lambda}_{n-1}}\\
x_{1}^{r_{1}}\ldots x_{n-1}^{r_{n-1}} & = & \hat{C}
\end{eqnarray*}

\noindent \emph{Define $B(x)=x_{1}^{\gamma_{1}}\ldots x_{n-1}^{\gamma_{n-1}}x_{n}$.
Then $A$ and $B$ have the same level surfaces and $g$ can be defined
accordingly. } 
\end{proof}
We are now ready to state a necessary and sufficient condition on
$A$ that ensures that its impermanent loss will satisfy the ERLI
condition. 
\begin{thm}
Let $A$ be an AMM. Then $A=g(B)$, where $g$ is a smooth function
of one real variable and $B$ is homogeneous in each coordinate if
and only if the formula for impermanent loss is ERLI for every liquidity
surface $A=k$. 
\end{thm}

\begin{proof}
\emph{We start with the forward direction. Fix $k$. Without loss
of generality, we can assume that $A$ is homogeneous in each coordinate
since $A=k$ corresponds to some $B=\tilde{k}$. By Lemma \ref{Lem 16},
$f$ is homogeneous in each of its $n-1$ coordinates, and so then
by Lemma \ref{Lem 15}, $W$ is homogeneous in each of its $n-1$
coordinates. Say that $W$ is homogeneous of degree $\alpha_{j}$
in $m_{j}$. Defining $t_{j}=\frac{m_{j}^{f}}{m_{j}^{i}}$ for $1\leq j\leq n-1$
and invoking Theorem \ref{Thm 14}, we have 
\begin{eqnarray*}
IL & = & \frac{W(m^{f})}{W(m^{i})+\sum_{j=1}^{n-1}(m_{j}^{f}-m_{j}^{i})W_{m_{j}}(m^{i})}-1\\
 & = & \frac{\prod_{j=1}^{n-1}\left[m_{j}^{i}\right]^{\alpha_{j}}}{\prod_{j=1}^{n-1}\left[m_{j}^{i}\right]^{\alpha_{j}}}\frac{W(t)}{W(1,\ldots,1)+\sum_{j=1}^{n-1}(t_{j}-1)W_{m_{j}}(1,\ldots,1)}-1\\
 & = & \frac{W(t)}{W(1,\ldots,1)+\sum_{j=1}^{n-1}(t_{j}-1)W_{m_{j}}(1,\ldots,1)}-1
\end{eqnarray*}
}

\noindent \emph{where in the second equality we have used the fact
that $W_{m_{j}}$ is homogeneous of degree $\alpha_{k}$ in $m_{k}$
if $k\neq j$ and is homogeneous of degree $\alpha_{j}-1$ in $m_{j}$.
Thus, noting that IL has been expressed entirely in terms of the exchange
rate ratios $t_{j}$, the forward direction is complete.}

\emph{Fix $k$ and consider the surface $A=k$ with induced function
$x_{n}=f(x_{1},\ldots,x_{n-1})$. For the reverse direction, by Lemma
\ref{Lem 17}, it suffices to show that $f$ is homogeneous in each
of its $n-1$ coordinates. By Lemma \ref{Lem 15}, it then suffices
to show that $W$ is homogeneous in each of its $n-1$ coordinates.
Choose a coordinate $m_{l}$ and fix $c\neq0$. Fix all other coordinates
aside from $m_{l}$ in defining the following functions:}

\noindent \emph{ 
\begin{eqnarray*}
A(m_{l}) & = & W(m_{1},\ldots,m_{l},\ldots,m_{n-1})\\
B(m_{l}) & = & W(m_{1},\ldots,cm_{l},\ldots,m_{n-1})
\end{eqnarray*}
}

\noindent \emph{The fact that impermanent loss is ERLI for the surface
$A=k$ implies that}

\noindent \emph{ 
\begin{eqnarray*}
\frac{A(m_{l}^{f})}{A(m_{l}^{i})+(m_{l}^{f}-m_{l}^{i})A'(m_{l}^{i})} & = & \frac{B(m_{l}^{f})}{B(m_{l}^{i})+(cm_{l}^{f}-cm_{l}^{i})A'(cm_{l}^{i})}\\
 & = & \frac{B(m_{l}^{f})}{B(m_{l}^{i})+(m_{l}^{f}-m_{l}^{i})B'(m_{l}^{i})}
\end{eqnarray*}
}

\noindent \emph{for all $m_{l}^{i},m_{l}^{f}\in\mathbb{R}$. Thus,
substituting $m_{l}^{f}=1$ and $m_{l}^{i}=s$ yields }

\emph{ 
\[
\frac{A(1)}{A(s)+(1-s)A'(s)}=\frac{B(1)}{B(s)+(1-s)B'(s)}
\]
which simplifies to 
\[
(s-1)\left[MA'(s)-B'(s)\right]=\left[MA(s)-B(s)\right]
\]
with $M=\frac{B(1)}{A(1)}$. Defining $C(s)=MA(s)-B(s)$, we obtain
the ODE 
\[
C'(s)=\frac{C(s)}{s-1}
\]
This is a simple separable ODE with solutions of the form 
\[
C(s)=E(s-1)
\]
for constant $E$}\footnote{\emph{$E$ might depend on $m_{j}$ for $j\neq l$.}}\emph{.
Note that even though the ODE has a singularity at $s=1$, from its
direction field, it is clear that the only differentiable functions
that satisfy it are of this form. Using the fact that the equation
$C(0)=MA(0)-B(0)=0$ holds, we conclude that $E=0$ and hence that}

\emph{ 
\begin{equation}
A(s)=MB(s)=MA(cs)
\end{equation}
}

\noindent \emph{The homogeneity of $W$ in $m_{l}$ will follow from
this equation, but there is a bit more work to do to see this. Since
$M$ depends on $c$, we have established that 
\[
A(cs)=g(c)A(s)
\]
for some function $g$. It turns out that this is enough to show that
$A(s)$ is homogeneous since $g$ satisfies the property $g(ab)=g(a)g(b)$:}

\noindent \emph{ 
\begin{eqnarray*}
A(abs) & = & g(ab)A(s)\\
 & = & g(a)A(bs)=g(a)g(b)A(s)
\end{eqnarray*}
}

\noindent \emph{Thus, $g$ is a solution to Cauchy's multiplicative
functional equation and this implies that $g(c)=c^{\gamma}$ for some
$\gamma$. For more background on such equations, consult \cite{Kuc09}.}

\emph{Note that we are still not quite done. We have shown that $A(cs)=c^{\gamma}A(s)$,
but our $\gamma$ may in theory depend on the $m_{j}$ for $j\neq l$.
To see that this is not the case, fix two vectors $(m_{1}^{a},\ldots,m_{l-1}^{a},m_{l+1}^{a},...,m_{n-1}^{a})$
and $(m_{1}^{b},\ldots,m_{l-1}^{b},m_{l+1}^{b},...,m_{n-1}^{b})$
in $\mathbb{R}^{n-2}$. We know that}

\noindent \emph{ 
\begin{eqnarray*}
W(m_{1}^{a},\ldots,cm_{l},\ldots,m_{n-1}^{a}) & = & c^{\gamma_{a}}W(m_{1}^{a},\ldots,m_{l},\ldots,m_{n-1}^{a})\\
W(m_{1}^{b},\ldots,cm_{l},\ldots,m_{n-1}^{b}) & = & c^{\gamma_{b}}W(m_{1}^{b},\ldots,m_{l},\ldots,m_{n-1}^{b})
\end{eqnarray*}
}

\noindent \emph{The fact that impermanent loss is ERLI for the surface
$A=k$ implies that}

\noindent \emph{ 
\[
\begin{array}{c}
\frac{c^{\gamma_{a}}W(m_{1}^{a},\ldots,1,\ldots,m_{n-1}^{a})}{W(m_{1}^{a},\ldots,1,\ldots,m_{n-1}^{a})+(c-1)W_{m_{l}}(m_{1}^{a},\ldots,1,\ldots,m_{n-1}^{a})}\\
=\frac{W(m_{1}^{a},\ldots,c,\ldots,m_{n-1}^{a})}{W(m_{1}^{a},\ldots,1,\ldots,m_{n-1}^{a})+(c-1)W_{m_{l}}(m_{1}^{a},\ldots,1,\ldots,m_{n-1}^{a})}\\
=\frac{W(m_{1}^{b},\ldots,c,\ldots,m_{n-1}^{b})}{W(m_{1}^{b},\ldots,1,\ldots,m_{n-1}^{b})+(c-1)W_{m_{l}}(m_{1}^{b},\ldots,1,\ldots,m_{n-1}^{b})}\\
=\frac{c^{\gamma_{b}}W(m_{1}^{b},\ldots,1,\ldots,m_{n-1}^{b})}{W(m_{1}^{b},\ldots,1,\ldots,m_{n-1}^{b})+(c-1)W_{m_{l}}(m_{1}^{b},\ldots,1,\ldots,m_{n-1}^{b})}
\end{array}
\]
}

\noindent \emph{Absorbing terms depending on $m_{j}^{a}$ and $m_{j}^{b}$
into constants yields}

\noindent \emph{ 
\[
\frac{c^{\gamma_{a}}K_{1}^{a}}{K_{2}^{a}+cK_{3}^{a}}=\frac{c^{\gamma_{b}}K_{1}^{b}}{K_{2}^{b}+cK_{3}^{b}}
\]
and once more simplified, an equation of the form}

\emph{ 
\[
J_{1}c^{\gamma_{a}+1}+J_{2}c^{\gamma_{a}}+J_{3}c^{\gamma_{b}+1}+J_{4}c^{\gamma_{b}}=0,
\]
where the $J$s and $K$s are constants that depend on the $m_{j}^{a}$s
and $m_{j}^{b}$s. By the strict concavity of $W$ (established by
the relation in Lemma \ref{Lem 15}), it is clear that $J_{1}$, $J_{2}$,
$J_{3}$, and $J_{4}$ are nonzero. As such, for the above equation
to hold for all $c$, $J_{1}=-J_{3}$, $J_{2}=-J_{4}$, and 
\[
\gamma_{a}=\gamma_{b}
\]
} 
\end{proof}
The following corollaries are useful when it is not immediately obvious
whether $A$ is of the form $g(B)$ with $B$ homogeneous in each
of its coordinates. 
\begin{cor}
Let $A$ be an AMM. Let $A=k$ be a liquidity surface and let $x_{n}=f(x_{1},\ldots,x_{n-1})$
be the function that it induces. Then $f$ is homogeneous in each
of its coordinates if and only if the impermanent loss for $A=k$
is ERLI. 
\end{cor}

\medskip{}

\begin{cor}
Geometric mean market makers and compositions of these market makers
with smooth real valued functions compose the space of ERLI market
makers. 
\end{cor}

\section{Conclusion}

In this paper, we have established a framework for reducing the parameters
involved in describing impermanent loss for automated market makers.
In doing so, we have shown that geometric mean market makers are the
simplest class of market makers from an impermanent loss standpoint.
More concretely, G3Ms exhibit a condition that we call \emph{Exchange
Rate Level Independence} (ERLI). ERLI is an interesting property in
that it is connected to how the dynamics of AMMs are affected when
token quantities are rescaled. Such rescalings can be done for either
conceptual or computational purposes. An example of this sort of rescaling
transformation can be found in the Curve V2 whitepaper \cite{Ego21}
and we believe that this mechanism also shifts Curve V2's impermanent
loss profile to a new profile that is closer to having ERLI. We will
work to quantify this in future work, and more generally, we are interested
in exploring how to extend our framework to dynamic AMMs.

\section{Acknowledgements}

The authors would like to thank Jamie Irvine for several interesting
conversations and for his insight regarding a critical lemma in this
paper. We would also like to thank Austin Pollok for sharing his insights
regarding parallels between DeFi and TradFi. Finally, we would like
to thank Guillermo Angeris for sharing his ideas and useful feedback
with us.

\newpage{}

\end{document}